\documentclass[10pt,bezier]{article}
\usepackage{epsfig,amsmath,amsthm,amsfonts,amssymb,verbatim}
\usepackage{pgf,tikz}
\usepackage{mathrsfs}
\usetikzlibrary{arrows}
\usepackage{authblk}
\usepackage{cite,fullpage}

\date{ }

\normalsize

\newtheorem{thm}{Theorem}
\newtheorem{cor}{Corollary}

\newtheorem{lem}{Lemma}

\newtheorem{conjecture}{Conjecture}
\theoremstyle{remark}
\newtheorem{remark}{\bf Remark}

\title{\bf Well-indumatched Trees and Graphs of Bounded Girth}
\author{\bf\small\sc S. Akbari$^{a}$
\footnote{Email addresses:
s$\_$akbari@sharif.edu (S. Akbari),
tinaz.ekim@boun.edu.tr (T. Ekim),
ghodrati84@gmail.com (A.H. Ghodrati),
sa$\_$zare$\_$f@yahoo.com (S. Zare).}, 
T. Ekim$^{b}$
A.H. Ghodrati$^{a}$, S. Zare$^{c}$\\
{\footnotesize {\em $^{\rm a}$Department of Mathematical Sciences,
Sharif University of Technology}}\\{\footnotesize {\em
$^{\rm b}$ Department of Industrial Engineering, Bogazici University}}\\{\footnotesize {\em
$^{\rm c}$Department of Mathematics and Computer
Science, Amirkabir University of Technology}}}

\begin{document}
\maketitle

\begin{abstract}
A graph $G$ is  called well-indumatched if all of its
maximal induced matchings have the same size. In this paper,
we characterize all well-indumatched trees.
We provide a linear time algorithm to decide if a tree is
well-indumatched or not.
Then, we characterize minimal well-indumatched graphs of girth at least 9 and show subsequently that for an odd integer $g\geq 9$ and $g\neq 11$, there is no well-indumatched graph of girth $g$. On the other hand, there are infinitely many well-indumatched unicyclic graphs of girth $k$,
where $k\in\{3,5,7\}$ or $k$ is an even integer greater than $2$. 
We also show that, although the recognition of well-indumatched graphs is 
known to be co-NP-complete in general, one can recognize in polynomial time 
well-indumatched graphs, where the size of maximal induced matchings is fixed.\\



\end{abstract}

\section{Introduction}\label{sec:intro}
Let $G$ be a graph.
A \emph{matching} in $G$ is a subset $M$ of $E(G)$ such that no two edges
of $M$ have a common endpoint. An \emph{induced matching} is a matching $M$
such that no two edges of $M$ is joined by an edge, in other words, 
$M$ occurs as an induced subgraph of $G$. In this paper, 
we are interested in graphs such that all their inclusion-wise maximal 
induced matchings have the same size. These graphs have been introduced very recently 
in \cite{Baptiste2017}, where they are called \emph{well-indumatched} graphs.

Induced matchings are also known as \emph{$2$-separated matchings} 
\cite{Stockmeyer1982}, \emph{strong matchings} \cite{Fricke1992} and 
\emph{distance-$2$ matchings}. This later is used in a more general context; 
a matching is called a \emph{distance-$k$} matching if the distance between any 
two edges of the matching is at least $k$, where the distance between two edges 
is the minimum of the distances (lengths of the shortest paths) between 
endvertices of the two edges. Clearly, an induced matching is a set of edges in 
which the distance between any two edges is at least $2$, hence a distance-2 matching.

Given a graph $G$,  {\sc Maximum Induced Matching}, called for short {\sc MIM}, is the problem of finding an induced matching of maximum size. This concept was introduced by Stockmeyer and 
Vazirani \cite{Stockmeyer1982}, where it was called ``risk-free marriage". As pointed out in \cite{Golumbic2000}, MIM finds applications in secure communication networks, VLSI design and network flow problems. Besides, it is closely related to    strong edge-colorings considered by Erd\"os  \cite{Erdos1988}, where every path of three edges needs three colours (see for instance \cite{Hocquard2013}). It follows from its definition that a strong edge-coloring of a graph boils down to partitioning its edge set into induced matchings. Another relation which makes MIM an interesting problem,
is its direct link with the irredundancy number of a graph \cite{Golumbic1993}. MIM is an important problem both for its applications and its relation to other important graph parameters. However, it is known to be NP-hard even for very restricted graph classes. For instance, it is known that MIM remains NP-hard in planar 3-regular graphs \cite{zito}, in planar bipartite graphs with degree $2$ vertices in one part and degree 
$3$ vertices in the other part \cite{shephard}, in $k$-regular bipartite graphs for any $k\geq 3$ \cite{Dabrowski2013}, and in Hamiltonian graphs, claw-free graphs and line graphs \cite{Kobler}. On the other hand, MIM is polynomial-time solvable in trees \cite{Golumbic2000}, chordal graphs \cite{Cameron1989}, circular arc graphs \cite{Golumbic1993} and interval graphs \cite{Golumbic2000}. One can refer to \cite{Dabrowski2013} and \cite{zito} for a more comprehensive literature review on the complexity status of MIM in various graph classes.

Another problem closely related to induced matchings, but less studied than MIM, is the problem of finding an inclusion-wise maximal induced matching of minimum size. This problem is called {\sc Minimum Maximal Induced Matching} and denoted by MMIM for short. MMIM has been shown to be NP-hard even in bipartite graphs of maximum degree 4 \cite{Orlovich2008} or in graphs having all of their maximal induced matchings of size either $k$ or $k+1$ for some integer $k\geq 1$ \cite{Baptiste2017}. The generalization of MMIM to distance-$k$ matchings has been also considered recently \cite{Kartynnik2017}.

When a graph parameter is hard to compute, one way to tackle this 
difficulty is to search for the so-called ``greedy instances",
where a greedy algorithm always ensures an optimal solution. 
In other words, one can be interested in graphs for which the difficult task of 
finding a largest or smallest set having a given property becomes trivial 
using a greedy approach. Several such structures have been studied in the 
literature. These results include \emph{well-covered} graphs defined as graphs 
such that all inclusion-wise maximal independent sets have 
the same size (see e.g. \cite{Plummer1993}), the edge analogue of 
well-covered graphs called \emph{equimatchable} graphs, where all maximal 
matchings have the same size (see e.g. \cite{Lesk1984}), \emph{well-dominated} 
graphs having all of their minimal dominating sets of the same 
size \cite{Finbow1988}, and  \emph{well-totally-dominated} graphs having all of 
their minimal total dominating sets of the same size \cite{Hartnell1997}.  

In \cite{Caro1996}, Caro, Seb\"o and Tarsi suggested a unified approach to study such greedy instances. Each one of the above mentioned graph classes has been extensively studied since then. For such a graph class $\mathcal G$, typical research questions considered in the literature include:
\begin{enumerate}
\item Structural characterizations of $\mathcal G$ and/or its subclasses,
\item Complexity of the recognition of the class $\mathcal G$ and/or its subclasses (usually obtained by the help of 1.),
\item Forbidden subgraphs in $\mathcal G$, if any, and characterization of hereditary graphs in $\mathcal{G}$, namely those graphs in $\mathcal{G}$ having all their induced subgraphs also in $\mathcal{G}$. 
\item Complexity of various graph problems in $\mathcal G$.
\end{enumerate}

As suggested in \cite{Caro1996}, 
the extensions of such ``greedy instances", where possible sizes of the sets with the 
desired property have only two possible (consecutive or not) values have also 
been considered in the case of equimatchability (the related graphs are called \emph{almost-equimatchable} \cite{Deniz2016}) 
and well-coveredness \cite{Ekim2018, Finbow1994}.

In the same spirit, a generalization of well-covered graphs, called $p$-equipackable graphs, were defined. A
graph is \emph{$p$-equipackable} if all its maximal $p$-packings are of the same size, where a \emph{$p$-packing} is a set of vertices such that the distance between any two distinct vertices in this set is larger than $p$ \cite{Kartynnik2017}. The edge analogue of $p$-equipackable graphs, called $p$-equimatchable graphs, were recently introduced in \cite{Kartynnik2017}; a graph is \emph{$p$-equimatchable} if all of its maximal distance-$p$ matchings have the same size. We note that equimatchable graphs are exactly 1-equimatchable graphs. Although deciding whether a given graph is equimatchable or not can be done in polynomial time \cite{Demange2014}, it has been shown in  \cite{Kartynnik2017} that the recognition of $p$-equimatchable graphs is co-NP-complete for any fixed $p\geq 2$. 
Note that 2-equimatchable graphs are exactly well-indumatched graphs which is the focus of our paper. In \cite{Baptiste2017}, it has been shown that recognizing a well-indumatched graph is a co-NP-complete problem even for $(2P_5, K_{1,5})$-free graphs. They also prove that, under the same restriction, the problem of recognizing a graph that has maximal induced matchings of at most $t$ distinct sizes is co-NP-complete for any given $t \geq 1$. After establishing the hardness of the recognition problem, the authors show that the decision versions of {\sc Independent Dominating Set}, {\sc Independent Set} and {\sc Dominating Set} problems are all NP-complete in the class of well-indumatched graphs. Then, they focus on the structure of well-indumatched graphs; they note that well-indumatched graphs are not hereditary and they characterize the so-called \emph{perfectly well-indumatched} graphs which are well-indumatched graphs such that all induced subgraphs are also well-indumatched. 

We start our paper with some definitions and preliminary results in Section \ref{sec:prem}. Then we proceed with the study of the structure and recognition of some subclasses of well-indumatched graphs. Note that for those graph classes $\mathcal G$, 
where both MIM and MMIM can be solved in polynomial time, one can decide whether a given graph in $\mathcal G$ is well-indumatched or not simply by solving each one of the two problems and checking if their optimal values coincide or not. However, in such a class $\mathcal G$, it is still interesting to find structural characterizations of well-indumatched graphs which can 
possibly lead to simpler recognition algorithms. This is the case for trees as both MIM and MMIM can be solved in linear time by the algorithms given in \cite{Golumbic2000} and \cite{Lepin2006}, respectively. In Section \ref{sec:trees}, we provide a simple characterization of well-indumatched trees which provides a much simpler linear time recognition algorithm. 

Section \ref{sec:unicyclic} is devoted to well-indumatched graphs with bounded girth. We characterize all minimal well-indumatched graphs of girth at least 9. This result implies that for an odd integer $g\geq 9$ and $g\neq 11$, there is no well-indumatched graph of girth $g$. On the other hand, there are infinitely many well-indumatched trees, 
infinitely many well-indumatched unicyclic graphs of girth $k$,
where $k\in\{3,5,7\}$ or $k$ is an even integer greater than $2$; 
and also infinitely many well-indumatched $r$-regular graphs of girth $3$,
where $r\geq 3$ is an arbitrary integer.

Finally, Section \ref{sec:fixed} is dedicated to well-indumatched graphs with a 
fixed size $k$ of  maximal induced matchings. We show that, when we consider the 
class of well-indumatched graphs with $k=1$, {\sc Weighted Independent Set} is 
polynomial time solvable, whereas {\sc Dominating Set} is NP-complete. 
This later result strengthens the known result of NP-hardness of 
{\sc Dominating Set} in well-indumatched graphs by restricting the size of 
maximal induced matchings to 1. Remind that the recognition of well-indumatched 
graphs is co-NP-complete even for $(2P_5, K_{1,5})$-free graphs. 
We show that deciding whether all induced matchings are of the same size $k$ can be done in polynomial time when $k$ is fixed.

\section{Definitions and Preliminaries}\label{sec:prem}

All graphs in this paper are finite, simple and undirected.
For a graph $G$, the vertex set and the edge set of $G$ are denoted
by $V(G)$ and $E(G)$, and their cardinalities $n$ and $m$ are called the
\emph{order} and the \emph{size} of $G$, respectively. The degree of a vertex $v\in V(G)$ is the number of vertices adjacent to it and it is denoted by $d(v)$. For an integer $r\geq 1$, a graph is said to be \emph{$r$-regular} if every vertex has degree $r$. The \emph{girth} of $G$ is the length of its shortest cycle.
The \emph{path} and the \emph{cycle} of order $n$
are denoted by $P_n$ and $C_n$, respectively. 
By $kH$ we denote the disjoint union of $k$ graphs each one isomorphic to $H$. 
We say that a graph $G$ is $H$-free, whenever $G$ does not contain $H$ as an 
induced subgraph. 

The \emph{distance} between two vertices $u,v\in V(G)$ is the
length of the shortest path between $u$ and $v$. 
The \emph{distance} between two edges $e_1, e_2$ of $G$ is defined as
the minimum distance between an end-point of $e_1$ and an end-point of $e_2$. 
For two edges $e_1$ and $e_2$,
we say that $e_1$ \emph{covers} $e_2$ if the distance between $e_1$ and $e_2$
is at most $1$. In particular, an edge covers itself. We say that a subset $F$ of edges covers an edge $e$ if there is an edge $f\in F$ such
that $f$ covers $e$. Note that an induced matching $M$ is maximal if and only if $M$ covers $E(G)$.

A graph $G$ is called \emph{reduced} if no two vertices of $G$ have the
same set of neighbors.
Note that a tree is reduced if and only if each vertex is adjacent to
at most one pendant edge (a \emph{pendant edge} is an edge incident with
a vertex of degree one).
Let $G$ be a graph.
For each set of vertices with same neighbors in $G$, remove
all of them except one.
Notice that this procedure will never create a new pair of vertices with the
same set of neighbors. Therefore, the resulting graph, called the \emph{reduction}
of $G$ and denoted by $R(G)$, is a uniquely defined reduced graph.

The following remark shows that one can restrict
the study of well-indumatched graphs to the reduced graphs.
\begin{remark}\label{rem:reduced}
The graph $G$ is well-indumatched
if and only if  $R(G)$ is  well-indumatched.
\end{remark}

\begin{proof}
To see this, assume without loss of generality that $x$ and $y$ are the only distinct vertices of $G$ having the same neighborhood. 
The general case can be shown by repeating the following argument for all such pairs. Any matching covering all edges incident with $x$ also covers all edges incident with $y$. Moreover, since an edge incident with $x$ covers all edges incident with $y$ and vice versa, an induced matching of $G$ can contain at most one edge incident with $x$ or $y$. So the result
follows.
\end{proof}

In \cite{Baptiste2017}, it is noticed that well-indumatched graphs are 
not hereditary as a $P_5$ is not well-indumatched but a $P_7$ which 
contains $P_5$ as an induced subgraph is well-indumatched. 
In addition, the authors provide a construction which shows that for any 
graph $H$, there is a well-indumatched graph $G$ containing $H$ as an 
induced subgraph. In other words, this certifies that there is no forbidden 
induced subgraph for a graph to be well-indumatched. 
Based on this observation, they characterize well-indumatched graphs which all 
their induced subgraphs are also well-indumatched, by three minimal forbidden 
subgraphs (see Theorem 10 in \cite{Baptiste2017}). 
Given a graph $G$, they also introduce the concept of 
\emph{co-indumatched subgraph} which is a subgraph $F$ of $G$ obtained by the 
removal of the closed neighborhood of the end-points of $M$ for some 
induced matching $M$ (possibly $M=\emptyset$) of $G$. 
This concept is then used to characterize well-indumatched graphs by 
forbidden co-indumatched subgraphs (Theorem 9 in \cite{Baptiste2017}). 
The following is a reformulation of Proposition 2 in \cite{Baptiste2017} 
which will be useful in our proofs.

\begin{lem}\label{lem:removal}
 Let $G$ be a well-indumatched graph and $F_0\subseteq E(G)$
 be an induced matching. If $F$ is the
 set of edges covered by $F_0$, then $G\setminus F$ is well-indumatched.
\end{lem}
\begin{proof}
Let $M_1$ and $M_2$ be two maximal induced matchings of $G\setminus F$. Then 
 $M_1\cup F_0$ and $M_2 \cup F_0$ are maximal induced matchings of $G$. 
 Since $G$ is well-indumatched, we have $|M_1\cup F_0| = |M_2\cup F_0|$, 
 which, together with the fact that $M_i \cap F_0 = \emptyset$, for $i=1,2$, 
 implies that $|M_1| = |M_2|$.
 Therefore $G\setminus F$ is well-indumatched.
\end{proof}

\section{Characterization of Well-Indumatched Trees}\label{sec:trees}

In this section we give a simple necessary and sufficient condition for a
reduced tree to be well-indumatched. Note that although this section is related to trees, Lemmas \ref{lem:cut-edge}, \ref{lem:P5} and \ref{lem:unique_cover} are for general well-indumatched graphs and can be useful in other contexts.

\begin{lem}\label{lem:longest_path}
	Let $T$ be a tree with a longest path $P=v_1\cdots v_k$ such that
	the degree of each vertex of $P$ is at most $3$, and those of degree $3$
	are incident to a pendant edge. 
	Then, $T$ is well-indumatched if and only if
	$k \in \{1,2,3,4\}$ or $k=7$ and $d(v_4)=2$.
\end{lem}
\begin{proof}
	It is not hard to see that if $k\in\{1,2,3,4\}$, 
	or $k=7$ and $d(v_4) = 2$, then $T$ is well-indumatched,
	and if $k \in \{5,6\}$ or $k=7$ and 
	$d(v_4)=3$, then $T$ is not well-indumatched.
	Now, let $k\geq 8$. Then,
	$T$ has an edge $e$ such that a component of the graph resulted from
	the removal of the edges covered by $e$ is
	a tree of the type described in the statement of the
	lemma and with a longest path of order $5$ or $6$.
	Since a graph is well-indumatched if and only if its connected
	components are well-indumatched, it follows from Lemma \ref{lem:removal} that $T$ is not well-indumatched.
\end{proof}

\begin{lem}\label{lem:indupath}
The following statements hold:\\
 (i) The path $P_n$ is well-indumatched if and only if $n\in\{1,2,3,4,7\}$.\\
 (ii) The cycle $C_n$ is well-indumatched if and only if
 $n\in\{3,4,\ldots,8,11\}$.
\end{lem}
\begin{proof}
 The first part follows from Lemma \ref{lem:longest_path}.
 For the second part, it can be seen that $C_n$ for $n\in\{3,4,\ldots,8,11\}$ is
 well-indumatched, and that $C_9$ and $C_{10}$ are not. If $n\geq 12$, then for any edge
 $e$ the removal of edges covered by $e$ results in a path with $n-4$ vertices
 which is not well-indumatched by the previous part.
 Therefore, by Lemma \ref{lem:removal}, $C_n$ is not well-indumatched.
\end{proof}

\begin{lem}\label{lem:cut-edge}
 Let $G$ be a connected well-indumatched  graph and $e=uv$ be a cut-edge
 of $G$. If $u$ is an end vertex of a path of length at least two
 which does not contain $v$, then the component of $G\setminus e$ containing
 $v$ is well-indumatched.
\end{lem}

\begin{proof}
 Let $H$ and $K$ be the components of $G\setminus e$
 containing $u$ and $v$, respectively. Let $uxy$ be a path in $H$.
 Extend $\{xy\}$ to a maximal induced matching $M$ in $H$.
 The edges covered by $M$ are $E(H)\cup \{e\}$, thus by Lemma \ref{lem:removal},
 $K=G \setminus (E(H)\cup \{e\})$ is well-indumatched.
 \end{proof}

The following describes two forbidden structures in a well-indumatched graph which will be useful in several proofs.

\begin{lem}\label{lem:P5}
Let $G$ be a graph. Then the following statements hold:
\begin{itemize}
 \item[(i)] If $v_1v_2v_3v_4v_5$ is a path in
 $G$ such that $d(v_1)=d(v_5)=1$ and $d(v_2)=2$,
 then $G$ is not well-indumatched.
 \item[(ii)] If $v_1v_2v_3v_4v_5v_6$ is a path in
 $G$ such that $d(v_1)=d(v_6)=1$ and $d(v_2)=d(v_5)=2$,
 then $G$ is not well-indumatched.
\end{itemize}
\end{lem}
\begin{proof}
(i) Note that every edge covered by $\{v_1v_2,v_4v_5\}$ is also covered by $v_3v_4$. Thus, if $M$ is a maximal
induced matching containing $v_3v_4$, then
$(M\setminus\{v_3v_4\})\cup \{v_1v_2,v_4v_5\}$
is an induced matching and is contained
in a maximal induced matching $M'$.
Clearly, $|M'|>|M|$, so $G$ is not well-indumatched. 

 (ii) Since every edge covered by $\{v_1v_2,v_5v_6\}$ is
 also covered by $v_3v_4$, we obtain the desired result exactly in the same manner as in item (i). 
\end{proof}

A pendant edge which is adjacent to a vertex of degree $2$ is called
a \emph{good pendant edge}.
\begin{lem}\label{lem:pendants_maximal}
  Let $T$ be a reduced well-indumatched tree of order at least $5$.
  Then the set of good pendant edges forms a maximal induced matching for $T$.
\end{lem}
\begin{proof}
  We prove the assertion by strong induction on $|V(T)|$. Since there is no tree of
  order $5$ which satisfies the hypothesis given in the statement of the lemma,
  the base step of the induction holds.
  Assume the result holds for all trees satisfying the hypothesis and with at most
  $n-1$ vertices, and consider a reduced well-indumatched tree $T$ of order $n$.
  Note that since $|V(T)| \geq 5$,
  the set of all good pendant edges forms an induced matching. 
  So, it suffices to show the set of all good pendant edges covers all
  edges of $T$.
  Let $v$ be a vertex of $T$ such that $d(v)\geq 3$. We prove the following two claims.
  
  \vspace{3mm}
  \noindent{\bf Claim 1.}
  Let $p$ be the number of pendant edges
  incident with $v$. Then at least $d(v)-1-p$ edges incident with $v$
  are contained in paths of length $3$ starting at $v$.
  (Note that since $T$ is reduced, $p=0$ or $p=1$.)
  
  {\it Proof of Claim $1$.} 
  By definition of $p$, $d(v)-p$ edges incident with $v$ are not 
  pendant edges and thus, are contained in paths of length $2$ starting at
  $v$. Suppose that two of these paths, say $vxx'$ and $vyy'$,
  are not contained in paths of length $3$ starting at $v$.
  This implies that $d(x')=d(y')=1$ and by Part (a) of Lemma \ref{lem:P5},
  $x$ should have a neighbor $z$ other than $v$ and $x'$. 
  Since $T$ is reduced,
  $z$ has a neighbor other than $x$, say $z'$.
  Now, $vxzz'$ is a path of length $3$, a contradiction.
  So, at most one of the $d(v)-p$ non-pendant edges incident with $v$
  is not contained in a path of length $3$ starting at $v$, and the
  claim follows.
  
  \vspace{3mm}
  \noindent{\bf Claim 2.} If there are two edges $vu_1$ and $vu_2$
  which are contained in paths of length $3$ starting at $v$, then
  the set of good pendant edges covers all edges of $T$,
  unless $d(v)=3$ and
  $v$ is incident with a pendant edge.
  
  {\it Proof of Claim $2$.} For $i=1,2$,
  let $T_i$ be the component of $T\setminus vu_i$ containing $v$.
  Then by Lemma \ref{lem:cut-edge}, $T_i$ is well-indumatched. Also, since $d(v) \geq 3$, $T_i$ is
  reduced and has at least $5$ vertices.
  Therefore, by the induction hypothesis,
  the set of good pendant edges of $T_i$,
  say $M_i$, covers all edges of $T_i$. Note
  that good pendant edges of $T_1$ and
  $T_2$ are good pendant edges of $T$, unless
  $d(v) = 3$ and $v$ is incident with a
  pendant edge.
  Since we have excluded the case $d(v)=3$ and $v$ is incident with a pendant edge
  in the statement of the Claim $2$,
  $M_1\cup M_2$ is
  a set of good pendant edges which  covers all edges of $T$ and the claim 
  is proved.
  
  \vspace{3mm}
  
  Now, let $v$ be an arbitrary vertex of $T$. If $d(v)\geq 4$, then
  by Claim $1$ there are two edges incident with $v$ which are contained
  in paths of length $3$ starting at $v$, and so Claim $2$ implies the result.
  Therefore, we can assume that for every $v\in V(T)$, $d(v)\leq 3$ and if $d(v)=3$, then
  either $v$ is incident with a pendant
  edge or there is at most one path of length $3$ starting
  at $v$. 
  
  Let $P=v_1v_2\cdots v_k$ be a longest path
  in $T$. Therefore $d(v_1)=d(v_k)=1$. If there is a vertex 
  $x\in V(T)\setminus V(P)$ such that $v_{k-1}x \in E(T)$,
  then since $T$ is reduced, there is a vertex
  $y\in V(T)\setminus(V(P) \cup \{x\})$ such that
  $xy \in E(T)$. Now, $v_1\cdots v_{k-1}xy$ is a path longer than $P$, which is impossible. 
  Thus, $d(v_{k-1})=2$ and similarly $d(v_2)=2$.
  By Lemma \ref{lem:P5}, there is no path of length $2$
  in $T\setminus E(P)$ starting at $v_3$ or $v_{k-2}$,
  and since $P$ is a longest path of $T$, there is no such path of 
  length at least $3$. 
  So $d(v_3),d(v_{n-2}) \in \{2,3\}$ and
  if $d(v_3) = 3$ (resp. $d(v_{n-2})=3$) then
  $v_2$ (resp. $v_{n-2}$) is incident with a pendant
  edge. Also, for 
  $3<i<n-2$, since there are two paths of length
  at least $3$ starting at $v_i$, either $d(v_i)=2$
  or $d(v_i)=3$ and $v_i$ is incident with a pendant edge.
  Therefore by Lemma \ref{lem:longest_path},
  $k\in \{1,2,3,4\}$ or $k=7$ and $d(v_4)=2$. In
  all of these cases, it is not hard to see that the set of good
  pendant edges covers all edges of $T$.	
\end{proof}

\begin{lem}\label{lem:unique_cover}
 Let $G$ be a graph and $M$ be a matching of $G$ which satisfies
 the following properties:

 i) Each edge of $G$ is covered by exactly one edge of $M$,

 ii) If $e_1,e_2\in E(G)$ are covered by $e_0\in M$, then
    $e_1$ covers $e_2$.\\
Then, $G$ is well-indumatched.
\end{lem}
\begin{proof}
 Let $M'$ be a maximal induced matching of $G$. We construct a one-to-one
 correspondence between $M'$ and $M$.
 Let $e\in M'$. By (i), there is exactly one $e_0\in M$ which covers $e$.
 Define $f(e)=e_0$. We claim that $f:M' \to M$ is a bijection from $M'$ onto
 $M$. If $e_1,e_2\in M'$ and $e_1\neq e_2$, but $f(e_1)=f(e_2)$,
 then by (ii), $e_1$ covers $e_2$ which contradicts
 the fact that $M'$ is an induced matching. So, $f$ is one-to-one. On the other
 hand, for each $e_0 \in M$, since $M'$ is a maximal induced matching,
 there is $e\in M'$ which covers $e_0$, so by the definition $f(e)=e_0$.
 Thus $f$ is onto.
 Therefore, the size of any maximal induced matching of $G$
 equals $|M|$ which implies that $G$ is well-indumatched.
\end{proof}

\begin{thm}\label{thm:tree}
 Let $T$ be a reduced tree of order at least $5$, and
 $M$ be the set of good pendant edges of $T$.
 Then, $T$ is well-indumatched if and only if
 each edge of $T$ is covered by exactly one edge of $M$.
\end{thm}

\begin{proof}
 Assume that every edge of $T$ is covered by exactly one edge of $M$. If $e=xy \in M$ and $z$ is the other neighbor of $y$, then the set of edges covered by $e$ consists of $e$
 and all edges incident with $z$. So the second
 condition of Lemma \ref{lem:unique_cover} also holds and therefore $T$ is well-indumatched.
 Conversely, let $T$ be well-indumatched. By Lemma \ref{lem:pendants_maximal},
 $M$ is a maximal induced matching of $T$, so every edge of $T$ is covered
 by at least one edge of $M$. If two edges of $M$ cover an edge of $T$, then
 $T$ has either a path $v_1v_2v_3v_4v_5$ such that $d(v_1)=d(v_5)=1$ and
 $d(v_2)=d(v_4)=2$ or a path $v_1v_2v_3v_4v_5v_6$ such that $d(v_1)=d(v_6)=1$ and
 $d(v_2)=d(v_5)=2$ and by Lemma \ref{lem:P5}, $T$ is not well-indumatched,
 a contradiction.
\end{proof}

In \cite{Golumbic2000}, a linear time algorithm for finding the maximum size of an
 induced matching in a tree is presented. Besides, a linear time algorithm is given for finding the size of a minimum maximal induced matching in a tree in \cite{Lepin2006}. These two algorithms provide a linear time
 algorithm for recognizing whether a tree is well-indumatched.
Our structural characterization of well-indumatched trees in Theorem \ref{thm:tree}, 
provides a more straightforward linear time recognition algorithm.

\begin{cor}
Given a tree $T$, it can be decided in linear time if $T$ is well-indumatched. 
\end{cor}
\begin{proof}
First, obtain $T'=R(T)$ in linear time. Let $M$ be the set of all good 
pendent edges of $T'$.  Clearly, $M$ can be formed in linear time just by 
checking the degrees of the parents of the leaves. Now, for each edge $e$ 
of $T'$, determine the edge(s) of $M$ covering $e$; this can be done by 
checking whether a neighbor of an end-point of $e$ appears in $V(M)$, 
thus requires a time proportional to the sum of the degrees, which is linear. 
By Theorem \ref{thm:tree}, $T'$ is well-indumatched if and only if each 
edge of $T'$ is covered exactly once. We conclude by Remark \ref{rem:reduced},
since $T$ is well-indumatched if and only if $T'$ is  well-indumatched. 
\end{proof}

\section{Well-indumatched Graphs of Bounded Girth}\label{sec:unicyclic}

In this section, we study well-indumatched graphs with lower bounded girth. 
We characterize all minimal well-indumatched graphs of girth at least 9. 
An important consequence of this characterization is that for an odd 
integer $g\geq 9$ and $g\neq 11$, there is no well-indumatched graph of 
girth $g$. It turns out that 
minimal well-indumatched graphs of girth at least 9 are (in particular) 
unicyclic. A \emph{unicyclic} graph is a connected graph with a unique cycle. 
Unlike for odd girth at least 9 (and not equal to 11), we show that there are 
infinitely many well-indumatched unicyclic graphs of even girth, or odd girth 
smaller than 9. We also show that there are infinitely many well-indumatched 
trees and infinitely many well-indumatched $r$-regular graphs of girth 3,
where $r\geq 3$ is an arbitrary integer. 

Let $G$ be a unicyclic graph and $C$ be the unique cycle of $G$.
For each $v\in V(C)$, the rooted tree in $G\setminus E(C)$ with root $v$
is denoted by $T_v$.
If $T$ is a rooted tree with root $v$, then the \emph{depth}
of $T$ is the longest length of a path starting at $v$. 
In Figure \ref{fig:tree_types}, two types of rooted trees which are
encountered in following results are shown. A graph $G$ is said to be {\em minimal} well-indumatched with property $\mathcal P$ if $G$ is a well-indumatched graph with property $\mathcal P$ and has no
proper well-indumatched subgraph with property $\mathcal P$.

\begin{lem}\label{lem:types}
If $G$ is a minimal well-indumatched graph of girth $g\geq 9$, then $G$ is a reduced unicyclic graph. Moreover, if the
unique cycle of $G$ is
$C = v_1\cdots v_g$, then every $T_{v_i}$ is of one of the Types (i) or (ii) in Figure \emph{\ref{fig:tree_types}}. 
\end{lem}

\begin{proof}
Let $C = v_1v_2\cdots v_gv_1$ be a cycle in $G$.
If there is $e\in E(G)$ covered by no edge
of $C$, then by removing $e$ and the edges covered by $e$,
we obtain a well-indumatched graph of girth $g$ with fewer edges than $G$,
contrary to the assumption. So every edge of $G$ is covered by some edge
of $C$. We claim that in $G\setminus E(C)$, 
the vertices $v_1,v_2,\ldots, v_g$ belong
to different components. Suppose not, and let $P$ be a path of minimum length
in $G\setminus E(C)$ between two vertices of $C$.
Let $v_i$ and $v_j$ be the
end vertices of $P$ and $uv_i$ and $wv_j$ be the
first and the last edges of $P$. Since $P$ has minimum length,
it has no common vertex
with $C$ except $v_i$ and $v_j$. 
Suppose that there is an edge $e=xy$ 
between $x \in V(P)$ and $y \in V(C)$ other than
$uv_i$ and $wv_j$. If $x \notin \{u,w\}$, then there is
a shorter path in $G\setminus E(C)$ between 
two vertices of $C$, which is impossible. Also, if 
$x=u$ or $x=w$, then $G$ has a cycle of size smaller than $g$.
So, there is no edge between $V(P)$ and $V(C)$ except that
first and the last edges of $P$. 
Let $Q$ be a shortest path on
$C$ between $v_i$ and $v_j$. Thus, the length of $Q$ is at most
$\lfloor \frac{g}{2} \rfloor$. Now, since $P\cup Q$ is a cycle,
the length of $P$ is at least
$\lceil \frac{g}{2} \rceil \geq 5$. But then $P$ has an edge which is not
covered by any edge of $C$, a contradiction.

So let $H_i$ be the component of $G\setminus E(C)$ containing
$v_i$, $1 \leq i \leq g$. If $H_i$ has a cycle, which should be of length at least $9$,
then $H_i$ has a path of length at least $7$ which does not contain $v_i$,
say $u_1,u_2\cdots u_8$. Since each edge of $G$ should be covered
by at least one edge of $C$, for each $j$, there is at least one edge
between $v_i$ and $\{u_j, u_{j+1}\}$,
which implies that $H_i$
has a cycle of length at most $4$ (by considering 3 consecutive edges 
$u_1u_2, u_2u_3, u_3u_4$), a contradiction to the girth assumption. 
Thus, each $H_i$ is a tree and $G$ is unicyclic.
	
If $G$ is not reduced, then by Remark \ref{rem:reduced}, $R(G)$ is a
proper well-indumatched subgraph of
 $G$ whose girth is $g$, a contradiction. So $G$ is reduced.
For some $v\in V(C)$,
 if the depth of $T_v$ is at least $3$ and $vx_1x_2x_3$ is a path in $T_v$,
 then by removing $x_2x_3$ and all edges covered by $x_2x_3$ we obtain
 a well-indumatched unicyclic graph with fewer edges, a contradiction.
 Thus, for each $v\in V(C)$ the depth of $T_v$ is at most $2$.
 
  Now, we prove that for
 each $i$, $T_{v_i}\neq P_2$. If $T_{v_i}=P_2$, for some $i$, then remove
 $v_{i+5}v_{i+6}$
 and all edges covered by $v_{i+5}v_{i+6}$. By Lemma \ref{lem:removal},
 the resulting graph is a well-indumatched forest.
 Let $T$ be the component of this forest containing $v_i$. By Theorem
 \ref{thm:tree},
 the single edge of $T_{v_i}$ should be covered by a good pendant edge
 of $R(T)$.
 If $g \geq 11$, then it is clearly impossible.
 Let $g=9$. For the single edge of $T_{v_i}$ to be
 covered by a good pendant edge of $R(T)$, it is
 necessary that $T_{v_{i-1}} = P_2$. By a similar argument,
 $T_{v_{i-2}} = P_2$, and continuing in this way, we conclude
 that $T_{v_j}$ is $P_2$ for all $j$. However, it is not hard to see that
 the resulting graph is not well-indumatched.
 Now, let $g=10$. For the single edge of $T_{v_i}$ to be
 covered by a good pendant edge of $R(T)$, it is necessary
 that
 $T_{v_{i-2}} = P_1$ and $T_{v_{i-1}} = P_1$ or $P_2$.
 Similarly, if one removes $v_{i-5}v_{i-6}$ and the edges covered by it, then it yields that $T_{v_{i+2}} = P_1$ and $T_{v_{i+1}} = P_1$ or $P_2$.
 Since $v_{i+1}v_{i+2}$ should be covered by a good
 pendant edge of $R(T)$, $T_{v_{i+3}} = P_2$
 (if $T_{v_{i+3}}=P_1$, then the edge $v_iv_{i+1}$ is covered by two good pendant edges of $R(T)$, which
 contradicts Theorem \ref{thm:tree}).
 By repeating this argument we conclude that 
 $T_{v_{i+3k}} = P_2$, for every $k$. So, each $T_{v_j}$ is $P_2$. However, it is not hard to see that the resulting graph is not well-indumatched, a contradiction. Now, the minimality of $G$ implies that for each $i$, $ T_{v_i}$ is of Type (i) or (ii) shown in Figure 2.

\end{proof}

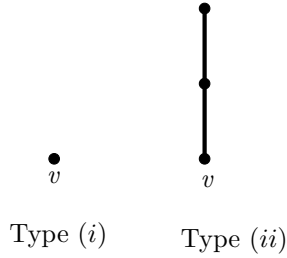
\begin{figure}
  \centering
  \begin{tikzpicture}[line cap=round,line join=round,>=triangle 45,x=1.0cm,y=1.0cm]
\clip(-1.,-2.) rectangle (4.,2.5);
\draw [line width=1.6pt] (2.,0.)-- (2.,1.);
\draw [line width=1.6pt] (2.,1.)-- (2.,2.);
\draw (-0.22,-0.04) node[anchor=north west] {${\bf \it v}$};
\draw (1.82,-0.1) node[anchor=north west] {${\bf \it v}$};
\draw (-0.72,-0.72) node[anchor=north west] {Type $(i)$};
\draw (1.56,-0.78) node[anchor=north west] {Type $(ii)$};
\begin{scriptsize}
\draw [fill=black] (0.,0.) circle (2.0pt);
\draw [fill=black] (2.,0.) circle (2.0pt);
\draw [fill=black] (2.,1.) circle (2.0pt);
\draw [fill=black] (2.,2.) circle (2.0pt);
\end{scriptsize}
\end{tikzpicture}
  \caption{Two types of rooted trees}\label{fig:tree_types}
\end{figure}

\begin{lem} \label{lem:girth9_11}The following statements hold:
\begin{itemize}
 \item[(i)] Let $G$ be a reduced well-indumatched unicyclic graph of girth $11$ with the unique cycle $C=v_1v_2\cdots v_{11}v_1$, such that for each $i$, $T_{v_i}$ is one of the two rooted trees shown
 in Figure \emph{\ref{fig:tree_types}}, then $G=C_{11}$.
  \item[(ii)]  There is no well-indumatched graph of girth $9$.
\end{itemize}
\end{lem}

\begin{proof}
(i) Consider the indices modulo 11.  Assume for a contradiction 
that there is an index $j$ such that $T_{v_j}\neq P_1$.
 
 We note that there is no index $i$ such that $T_{v_i}$ and
 $T_{v_{i+1}}$ are of Type (ii), because this would induce the forbidden structure in Lemma \ref{lem:P5} (ii). It follows that for each $i$, at least one
 of $T_{v_i}$ and $T_{v_{i+1}}$ is of Type (i).
 Now, $T_{v_j}$ is of Type (ii), $T_{v_{j+1}}$
 is of Type (i). By removing $v_{j-3}v_{j-2}$ and the edges covered by it,
 we obtain a well-indumatched forest. If $T$ is the component of this forest
 which contains $v_{j+1}v_{j+2}$, then $R(T)$ has at least 5 vertices 
 $v_j, \ldots ,v_{j+4}$ and by Theorem \ref{thm:tree}, 
 every edge of $R(T)$, in particular edge $v_{j+1}v_{j+2}$ should be 
 covered by exactly one good pendant edge of $R(T)$. 
 This implies that $T_{v_{j+2}}$ is of Type (ii),
 and therefore $T_{v_{j+3}}$ is of Type (i). Continuing in this way, we conclude that
 $T_{v_{j+5}},T_{v_{j+7}},T_{v_{j+9}},T_{v_{j+11}} =T_{v_j}$ is of Type (i), a contradiction. Thus $G=C_{11}$.

(ii) 
Assume that $G$ is a minimal well-indumatched graph of girth $9$. 
Then, by Lemma \ref{lem:types}, $G$ is a reduced unicyclic graph. 
Moreover, if $C=v_1v_2\cdots v_{9}v_1$ is its unique cycle, then for each 
index $i$, $T_{v_i}$ is one of the Types (i) or (ii) shown in Figure
 \ref{fig:tree_types}. By Lemma \ref{lem:indupath}, we know that $C_9$ is 
 not well-indumatched. So, assume that there is an index $j$ such 
 that $T_{v_j}\neq P_1$. 
 Now, take the indices modulo $9$. 
 As in Part (i), there is no index $i$ such that $T_{v_i}$ and
 $T_{v_{i+1}}$ are of Type (ii), because this would induce the 
 forbidden structure in Lemma \ref{lem:P5} (ii). 
 So, for each $i$, at least one
 of $T_{v_i}$ and $T_{v_{i+1}}$ is of Type (i).
 Since $T_{v_j}$ is of Type (ii), $T_{v_{j+1}}$
 is of Type (i). By removing $v_{j-3}v_{j-2}$ and the edges covered by it,
 we obtain a well-indumatched forest. If $T$ is the component of this forest
 which contains $v_{j+1}v_{j+2}$, then $R(T)$ has at least $5$ vertices 
 $v_j, \ldots ,v_{j+4}$ and by Theorem \ref{thm:tree}, 
 every edge of $R(T)$, in particular edge $v_{j+1}v_{j+2}$ should be 
 covered by exactly one good pendant edge of $R(T)$.
 Since the girth is 9, we note that the unique good 
 pendent edge of $R(T)$ covering $v_{j+1}v_{j+2}$ can also be the edge 
 $v_{j+3}v_{j+4}$, unlike for girth $11$. 
 If $v_{j+3}v_{j+4}$ is the unique good pendant edge of $R(T)$
 which covers $v_{j+1}v_{j+2}$,
 then $T_{v_{j+4}}, T_{v_{j+3}}, T_{v_{j+2}}, T_{v_{j+1}}$  are all of Type (i). 
 Now, remove edge $v_{j-4}v_{j-3}$ and all edges covered by it. 
 This leaves a well-indumatched forest, let $T'$ be its component containing 
 $v_{j+3}$. 
 Let also $x_1x_2v_{j}$ be the path of length $2$ in $T_{v_j}$. 
 Then $x_1,x_2,v_{j},v_{j+1},v_{j+2}, v_{j+3}$ induce a $P_6$ which is 
 forbidden for being well-indumatched by Lemma \ref{lem:P5} (ii), a contradiction. 
 Therefore, edge $v_{j+1}v_{j+2}$ is covered by a good pendent edge in 
 $T_{v_{j+2}}$ which is of Type (ii), and we obtain a contradiction as for 
 girth $11$. We conclude that there is no well-indumatched graph of girth 9. 
 \end{proof}

\begin{lem}\label{lem:unicyclic_Tv}
 Let $G\neq C_{11}$ be a reduced well-indumatched unicyclic graph of girth $g \geq 10$
 with the unique cycle $C=v_1v_2\cdots v_gv_1$,
 such that for each $i$, $T_{v_i}$ is one of the two rooted trees shown
 in Figure \ref{fig:tree_types}. Then, for each $i$ (modulo 9), the trees $T_{v_i}$
 and $T_{v_{i+1}}$ are alternately of Type (i) and (ii). In particular, $g$ is even and the size of any maximal induced matching is $\frac{g}{2}$.
\end{lem}
\begin{proof}
Let $i$ be an integer, $1 \leq i \leq g$.
 By removing the edge $v_{i+6}v_{i+7}$ and
 the edges covered by $v_{i+6}v_{i+7}$, we obtain a well-indumatched forest.
 Let $T$ be the component of that forest containing $v_iv_{i+1}$. Since $T$ 
 contains a path of order $6$ induced by 
 $v_{i-1},\ldots, v_{i+4}$, $R(T)$ is a reduced well-indumatched tree with at 
 least $5$ vertices. Then, by
 Theorem \ref{thm:tree} every edge of $R(T)$ is covered by exactly
 one of the good pendant edges of $R(T)$. In particular, the edge $v_iv_{i+1}$ 
 has to be covered by exactly one good pendent edge. If $g\geq 12$, it follows 
 that exactly one of $T_{v_i}$ and $T_{v_{i+1}}$ should be of Type (ii) and the 
 other of Type (i). This yields that  for each $i$ modulo 9, the trees $T_{v_i}$
 and $T_{v_{i+1}}$ are alternately of Type (i) and (ii) and consequently $g$ is even. 
 By Lemma \ref{lem:girth9_11} (i), $C_{11}$ is the only graph of girth 11 with 
 the desired properties, thus $g\neq 11$. If $g=10$ and $v_{i-1}v_i$ is a good 
 pendent edge of $R(T)$ covering $v_iv_{i+1}$ (thus also $v_{i+1}v_{i+2}$), 
 then $T_{v_{i-1}}, T_{v_{i}}, T_{v_{i+1}}, T_{v_{i+2}}$ are of Type (i) and 
 $v_{i+3}v_{i+4}$ is not a good pendent edge of $R(T)$. Since $v_{i+2}v_{i+3}$ 
 should be covered by one good pendent edge, $T_{v_{i+3}}$ is of Type (ii) and 
 $T_{v_{i+4}}$ is of Type (i). Now, removing the edge $v_{i+7}v_{i+8}$ and
 the edges covered by $v_{i+7}v_{i+8}$, we obtain a well-indumatched forest. 
 However, its component containing $v_{i+3}$ has a $P_6$ induced by 
 $v_{i}, v_{i+1}, v_{i+2}$ and three vertices of $T_{v_{i+3}}$ which is 
 forbidden for being well-indumatched by Lemma \ref{lem:P5} (ii), a contradiction. 
 It follows that, in order to cover the edge $v_iv_{i+1}$ by exactly one good 
 pendent edge of $R(T)$, exactly one of $T_{v_i}$ and $T_{v_{i+1}}$ should be of 
 Type (ii), and the other of Type (i). Then, we conclude as in the case $g\geq 12$.

We complete the proof by noting that the set of good pendent edges of $G$ forms an induced matching of size $\frac{g}{2}$. Moreover, it covers all edges, thus it is maximal. Since $G$ is well-indumatched, it follows that all maximal induced matchings have size 
$\frac{g}{2}$.
\end{proof}

\begin{cor}
The only well-indumatched unicyclic graph of girth $11$ is $C_{11}$.
\end{cor}
We are now ready to characterize all minimal well-indumatched graphs of girth at least 9.

\begin{thm}\label{thm:minimalWIM}
The graph $G$ is a minimal well-indumatched graph of girth at least 9 if and only if either $G=C_{11}$ or $G$ is a reduced  unicyclic graph of even girth $g \geq 10$ with the unique cycle $C=v_1v_2\cdots v_gv_1$,
 such that for each $i$, $T_{v_i}$ is alternately of Type (i) and Type (ii).
\end{thm}
\begin{proof}
Let $G$ be a minimal well-indumatched graph of girth $g\geq 9$. By Lemma \ref{lem:types}, $G$ is a reduced unicyclic graph. Moreover, if $C=v_1v_2\cdots v_gv_1$ is its unique cycle, then every $T_{v_i}$ is of one of the Types (i) or (ii).  By Lemma \ref{lem:girth9_11} (ii), there is no well-indumatched graph of girth $9$. So, we have $g(G)\geq 10$. Then, either $G=C_{11}$ by Lemma \ref{lem:girth9_11} (i) or the length of the unique cycle is even and for each $i$, $T_{v_i}$ is alternately of Type (i) and Type (ii) by Lemma \ref{lem:unicyclic_Tv}.

Let now $G$ be a reduced  unicyclic graph of even girth $g$ at least 10 with the 
unique cycle $C=v_1v_2\cdots v_gv_1$, such that for each $i$, $T_{v_i}$ is 
alternately of Type (i) and Type (ii). Let also $M$ be a matching of $G$ 
containing all good pendent edges. Then, each edge of $G$ is covered by exactly 
one edge of $M$. Moreover, if two edges of $G$ are covered by the same edge of $M$, 
then these two edges cover each other. It follows from 
Lemma \ref{lem:unique_cover} that $G$ is well-indumatched. 
Let us now show that $G$ is also a minimal well-indumatched graph of girth $g$. 
Indeed, any subgraph $G'$ of $G$ of the same girth $g$ is obtained by removing 
(at least one) edge from $E(G)\setminus E(C)$. 
If $R(G')$ has $T_{v_i}=P_2$ for some $i$, then by Lemma \ref{lem:types}, 
$G'$ is not well-indumatched. 
So assume every $T_{v_i}$ in $R(G')$ is of  one of the Types (i) or (ii) with 
both $v_i$ and $v_{i+1}$ of degree 2 (in $R(G')$), for some $i$. 
Then, $G'$ is not well-indumatched by Lemma \ref{lem:unicyclic_Tv}.
\end{proof}

It is not hard to see that given a graph $G$, one can check in time $O(m)$ whether $G$ is reduced and unicyclic; and if it is unicyclic, whether for each $i$, $T_{v_i}$ is alternately of Type (i) and (ii) (and therefore whether it has even girth at least 10). It follows from Theorem \ref{thm:minimalWIM} that one can decide in time $O(m)$ whether a given graph is minimal well-indumatched of girth at least 9. However, unfortunately, this result does not imply a polynomial time algorithm to recognize well-indumatched graphs of girth at least 9. Indeed, this is due to the fact that the property of being well-indumatched is not hereditary, as we already noted in Section \ref{sec:prem}. 

Theorem \ref{thm:minimalWIM} also implies the following, which is of interest by its own.

\begin{cor}\label{cor:unicyclic_even}
 For an odd integer $g\geq 9$ and $g\neq 11$, there is no well-indumatched graph of girth $g$.
\end{cor}

Unlike this negative result, we show in what follows that there are infinitely many well-indumatched unicyclic graphs of even girth or small odd girth, that is 3, 5 and 7.   

Let $r \geq 1$ and $k \geq 0$ be integers and
$S_{r,k}$ be a tree obtained by subdividing each edge of $K_{1,r}$
 by $k$ vertices. 
Consider the disjoint union of $C_3$ and
$S_{r,2}$. Join the vertex of degree $r$ in $S_{r,2}$
to a vertex of $C_3$ and add a new vertex and join it
to another vertex of $C_3$. Denote this graph by $G_r$ and
note that the order of this graph $3r+5$.
Consider the disjoint copy of $C_5$ and
$S_{r,2}$ and identify the vertex of degree $r$ in $S_{r,2}$
with a vertex of $C_5$ and  denote the resulting graph by $H_r$.
The order of this graph is $3r+5$.
Also, identify a vertex of $C_7$
and the vertex of degree $r$ in $S_{r,3}$ and add a new vertex
and join it to a neighbor of the identified vertex in $C_7$.
Denote this graph by $L_r$. The order of $L_r$ is 
$4r+8$.
Finally, for any even integer $k\geq 4$, consider the disjoint copy of 
$C_k$  with the vertex set $\{v_1,\ldots,v_k\}$ and
$S_{r,2}$ and identify the vertex of degree $r$ in $S_{r,2}$
with $v_1$. Also, add $\frac{k}{2}$ copies of
$P_2$ and join them to $v_2,v_4,\ldots,v_{k}$. We denote the resulting
graph by $Q_{k,r}$ (see Figure \ref{fig:girth3578}).

\begin{figure}
    \centering
    \includegraphics[width=14cm]{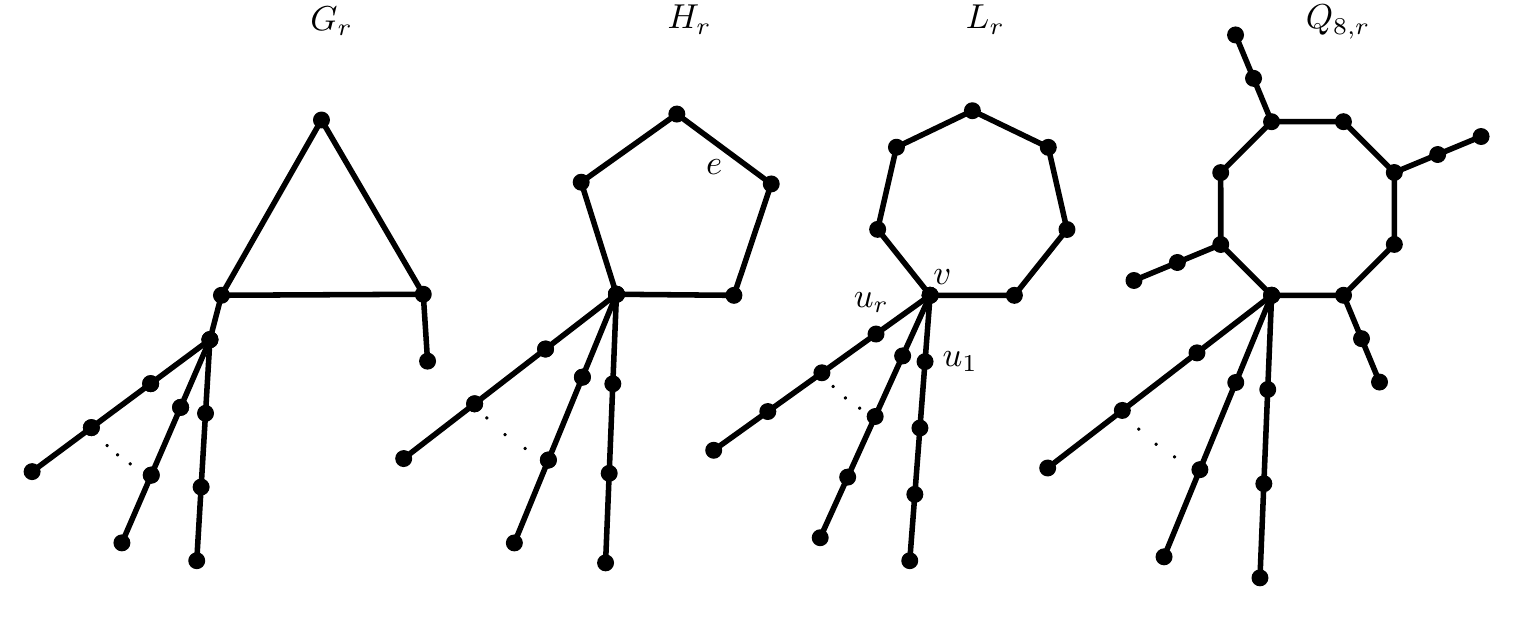}
    \caption{Well-indumatched unicycle graphs of girths $3,5,7$ and $8$}
    \label{fig:girth3578}
\end{figure}

\begin{thm}
 Let $r$ be a positive integer. Then the following statements
 hold:
 \begin{itemize}
 \item[(i)] $S_{r,2}$ is a well-indumatched tree.

 \item[(ii)] $G_r$ is a well-indumatched unicyclic graph
 of girth $3$.
 
 \item[(iii)] $H_r$ is a well-indumatched unicyclic graph
 of girth $5$.
  
 \item[(iv)] $L_r$ is a well-indumatched unicyclic graph
 of girth $7$.
 
 \item[(v)] For every even integer $k\geq 4$, $Q_{k,r}$
 is a well-indumatched unicyclic graph of girth $k$.
\end{itemize}
Thus, there are infinitely many well-indumatched trees 
and infinitely many well-indumatched unicyclic graphs of girth $k$,
where $k\in\{3,5,7\}$ or $k$ is an even integer greater than $2$.

\end{thm}
\begin{proof}
 (i, ii, v) Let $G$ be one of the
 graphs $S_{r,2}$, $G_r$ or $Q_{k,r}$.
 The pendant edges of $G$
 form a matching and each edge of $G$ is covered by exactly one of
 these edges. Moreover, any two edges of $G$ which are
  covered by the same pendent edge, cover each other. Therefore, Lemma \ref{lem:unique_cover} implies that $G$ 
 is well-indumatched.
 
 (iii) Let $M$ be the set of pendant edges of $H_r$ 
 together with $e$ (see Figure \ref{fig:girth3578}). 
 $M$ is a matching of $H_r$ and
 each edge of $H_r$ is covered by exactly one edge of $M$. Moreover, any two edges of $H_r$ which are covered by the same edge of $M$, cover each other. Therefore,  by Lemma
 \ref{lem:unique_cover}, $H_r$ is well-indumatched.
 
 (iv)
 For $i=1,\ldots,r$, denote the path of length $4$ in $L_r$ which starts at $v$ 
 and contains $u_i$ by $P_{u_i}$.
 Let $M$ be a maximal induced matching of $L_r$.
 If $M$ contains $vu_i$, for some $i$, then $M$ should contain
 two edges of $P_{u_i}$, one edge of each $P_{u_j}$, for $j \neq i$,
 and one edge of $C_7$. So, $|M|=r+2$. If $M$ does not contain $vu_i$ for any $i$,
 then $M$ should contain one edge of each $P_{u_j}$ and
 two edges of $C_7$. Thus, $|M|=r+2$.
 Therefore, the size of every maximal induced matching of $L_r$ is $r+2$
 and $L_r$ is well-indumatched.
\end{proof}
The following result gives an infinite family of well-indumatched graphs of girth 3, with the additional property of being regular.
\begin{thm}
 For every positive integer $r\geq 3$, there are infinitely many
 well-indumatched $r$-regular graphs of girth $3$.
\end{thm}
\begin{proof}
 Let $t$ be a positive integer and consider $t$ disjoint copies of complete graphs of order $r+1$,
 say $K_{r+1}^1,K_{r+1}^2,\ldots,K_{r+1}^t$
 and let $x_iy_i,u_iv_i\in E(K_{r+1}^i)$ be two disjoint edges,
 $i=1,\ldots,t$.
 Now, join $u_i$ to $x_{i+1}$ and $v_i$ to $y_{i+1}$, 
 for $i=1,\ldots,t-1$, then remove $x_iy_i$
 and $u_iv_i$, for $i=1,\ldots,t$, and call the resulting graph $G_t$.
 Also, let $H_t = G_t+x_1y_1$ and $L_t = G_t+x_1y_1+u_tv_t$.
 Note that $L_t$ is $r$-regular and has girth $3$.
 By strong induction on $t$, we prove that the size of all maximal induced matchings
 of $G_t$, $H_t$ and $L_t$ is $t$ and therefore, these graphs are 
 well-indumatched. 
 For $t=1$ the assertion is clear. Let the result hold for $G_j,H_j,L_j$,
 $j=1,\ldots,t-1$.
 Let $M$ be a maximal induced matching of $L_t$. If for all $i$, $1\leq i \leq t-1$, none of the $u_ix_{i+1}$ and
 $v_iy_{i+1}$ are in $M$, then it is not hard to see
 that $|M|=t$. 
 If there is
 some $i$, $1\leq i \leq t-1$, such that one of the $u_ix_{i+1}$ or
 $v_iy_{i+1}$ is in $M$, then both of these edges should be in $M$. By removing
 the edges covered by $u_ix_{i+1}, v_iy_{i+1}$, we obtain
 a graph with two components which are isomorphic to $G_{i-1}$ and $G_{t-i-1}$.
 Since the restriction of $M$ to each of these graphs is a maximal 
 induced matching, by the induction hypothesis, $M$ has $i-1$ edges
 in the component isomorphic to $G_{i-1}$ and $t-i-1$ edges
 in the other component. Thus, $|M|=2+(i-1)+(t-i-1) = t$. The proof
 for $H_t$ and $G_t$ is similar.
\end{proof}

We conclude this section by noting that our results settle the existence of well-indumatched graphs with all possible girths except girth 11. In other words, for all girth $g\neq 11$ at least 3, either we establish that there is no well-indumatched graph of girth $g$ or we exhibit an infinite family of well-indumached graphs of girth $g$. The only exception  of this dichotomy is $g=11$ for which we only know that the only minimal well-indumatched graph of girth 11 is $C_{11}$.  We conjecture the following:
\begin{conjecture}
The cycle $C_{11}$ is the only
connected well-indumatched graph of girth 11.
\end{conjecture}

\section{Well-indumatched graphs with maximal induced matchings of fixed size}\label{sec:fixed}

Let us call a graph \emph{$k$-well-indumatched} if all of its maximal 
induced matchings have size $k$.  
In this section, we focus on $k$-well-indumatched graphs with fixed $k$. 
Let $K_n$ be a clique on $n$ vertices. We start with a general observation:

\begin{remark}\label{rem:K2}
	If $G$ is a $k$-well-indumatched graph then it is $(k+1)K_2$-free.
\end{remark}
On the other hand, as expected, if $\ell$ is the smallest integer such that $G$ is $\ell K_2$-free, then $G$ is not necessarily $(\ell-1)$-well-indumatched. For instance, a $P_5$ is $3K_2$-free, but it is not 2-well-indumatched. However, 
if we restrict to 1-well-indumatched graphs, then the converse becomes true as 
well. Indeed, if $G$ is not $1$-well-indumatched, then $G$ is either 
$k$-well-indumatched with $k\geq 2$, or it is not well-indumatched. 
In both cases, $G$ has an induced matching of size 2 which induces a $2K_2$. 
So, we have the following:

\begin{remark}\label{prop:1WIM}
	A graph $G$ is $1$-well-indumatched if and only if $G$ is a non-empty $2K_2$-free graph.
\end{remark}

Remark \ref{prop:1WIM} is a forbidden subgraph characterization for $1$-well-indumatched graphs. This implies directly that 1-well-indumatched graphs form a hereditary class of graphs, that is a class of graphs closed under taking induced subgraphs. Note that this is in contrast with the non-hereditary nature of $k$-well-indumatched graphs starting from already $k\geq 2$; recall the example of a $P_7$ which is 2-well-indumatched but contains a $P_5$ which is not well-indumatched. 

Remarks \ref{rem:K2} and \ref{prop:1WIM} have some consequences from the computational complexity point of view. Note that {\sc Independent Set} and {\sc Dominating Set} problems are shown to be NP-complete in the class of well-indumatched graphs \cite{Baptiste2017}. In contrast to this hardness result, Remark \ref{rem:K2} and the fact that {\sc Weighted Independent Set} is polynomial time solvable in $kK_2$-free graphs, for any fixed $k$ \cite{Balas1989}, implies the following:
\begin{cor}
{\sc Weighted Independent Set} is polynomial time solvable,
when restricted to $k$-well-indumatched graphs for any positive $k$. 
\end{cor}
On the other hand, the NP-completeness of the  {\sc Dominating Set} problem in well-indumatched graphs can be strengthened using Remark \ref{prop:1WIM} and the fact that  {\sc Dominating Set} is NP-complete in split graphs \cite{Bertossi1984, Corneil1984}. A graph is {\em split},
if its vertex set can be partitioned into a clique and an independent set. It is known that $G$ is a split graph, if and only if it contains no $2K_2,C_4$ or $C_5$ as induced subgraph \cite{golumbic}. Thus, split graphs are $2K_2$-free and therefore 1-well-indumatched graphs contain the class of split graphs. This implies that {\sc Dominating Set} is NP-complete in well-indumatched graphs, even if every maximal induced matching has size 1.
\begin{cor}\label{cor:1WIM}
	{\sc Dominating Set} is NP-complete in 1-well-indumatched graphs.
\end{cor}
Note that there is no containment relationship between $k$-well-indumatched graphs and $k'$-well-indumatched graphs for $k'>k$ (they are actually disjoint sets partitioning the set of all well-indumatched graphs) and therefore, Corollary \ref{cor:1WIM} has no consequence on the NP-completeness of {\sc Dominating Set} in $k$-well-indumatched graphs for $k>1$.  

Remark \ref{prop:1WIM} is also an important intermediary result which makes the recognition of $k$-well-indumatched graphs polynomial time solvable, whenever $k$ is fixed. Note that this is in contrast with the co-NP-completeness of the recognition problem in general \cite{Baptiste2017}.

\begin{thm}\label{thm:recognition}
    Given a graph $G$, it can be decided in time $O(m^{k-1}n^4)$ whether $G$ is $k$-well-indumatched or not,
    where $n$ and $m$ are the order and the size of $G$,
    respectively.
\end{thm}
\begin{proof}
We note that a graph $G$ is $k$-well-indumatched if and only if for every 
edge $e\in E(G)$, the graph $G \setminus C(e)$, where $C(e)$ is the set of edges 
covered by $e$, is $(k-1)$-well-indumatched. This is indeed a necessary condition 
for $G$ being $k$-well-indumatched. Besides, if for all $e\in E(G)$, 
every maximal induced matching of $G\setminus C(e)$ has size $k-1$, 
then every maximal induced matching of $G$ has size $k$, 
thus $G$ is $k$-well-indumatched. Now, repeat recursively $k-1$  times the 
removal of an edge $e$ together with $C(e)$. The above equivalence implies 
that $G$ is well-indumatched if and only if this recursive  procedure yields a 
1-well-indumatched graph for any choice of $k-1$ edges throughout the recursive 
procedure. Since there are at most $O(m^{k-1})$ such choices and whether the 
remaining graph is 1-well-indumatched or not can be checked in time $O(n^4)$ 
(by Remark \ref{prop:1WIM}, simply by checking if any possible subset of four 
vertices induces a $2K_2$), the overall procedure takes time $O(m^{k-1}n^4)$.
\end{proof}

\section{Conclusion}
Well-indumatched graphs were introduced to the literature very recently. Consequently, the structure of well-indumatched graphs is not yet well understood and seems to be a very promising research area. In this work, we characterized well-indumatched trees and studied well-indumatched graphs of bounded girth. We established several structural results on well-indumatched graphs of bounded girth and conjectured that there is no connected well-indumatched graph of girth 11 other than $C_{11}$. 

As a future research, it would be interesting to characterize those well-indumatched graphs in special graph classes and to derive polynomial time recognition algorithms. Our characterization of minimal well-indumatched graphs of girth at least 9 (in Theorem \ref{thm:minimalWIM}) do not seem to imply directly a polynomial time recognition algorithm for well-indumatched graphs of girth at least 9. It would be interesting to exploit this characterization in order to develop such a recognition algorithm, or to investigate the recognition of well-indumatched graphs of bounded girth more broadly. Some other graph classes that could be investigated in this direction are interval graphs, claw-free graphs or equimatchable graphs. 

Another research direction would be the study of graphs having a bounded gap (1 or some fixed $k$) between the size of a maximum induced matching and minimum maximal induced matching.  This approach has been applied to well-covered graphs and yielded several significant results (see e.g. \cite{Barbosa2013, Ekim2018}), and more recently to equimatchable graphs \cite{Deniz2016}.

\section{Acknowledgement}
The work of the second author is supported by the Turkish Academy of Science GEBIP award. Part of her research was carried out during her stay at the University of Oregon Computer and Information Science Department under Fulbright Association Visiting Scholar Grant and TUBITAK 2219 Programme, all of whose support is greatly appreciated.
Also, the research of the first, third and fourth authors was partly funded by Iran National
Science Foundation (INSF) under the contract No. 96004167.
Also, the last author was partly funded by Iran National
Science Foundation (INSF) under the contract No. 93030963.

\end{document}